\documentclass[a4paper,10pt]{article}
\usepackage{amsfonts}
\usepackage{amssymb}
\usepackage{latexsym}
\usepackage{enumerate}
\usepackage{amsmath}
\usepackage{amsthm}
\usepackage{graphicx}
\usepackage{multirow}
\usepackage{bbm}
\usepackage[export]{adjustbox}
\usepackage[numbers, comma, sort&compress]{natbib}
\usepackage{color}
\usepackage[title]{appendix}

\newtheorem{proposition}{Proposition}

\begin{document}
\bibliographystyle{apalike}

\title{Stratified communities in complex business networks}
\author{Roy Cerqueti$^{\flat}$\thanks{Corresponding author.}, Gian Paolo Clemente$^{\natural}$,  Rosanna Grassi$^{\sharp}$\\
 $^{\flat}$ {\small University of Macerata, Department of Economics and Law. }\\{\small Via Crescimbeni 20, 62100, Macerata,
Italy.}\\ {\small Tel.: +39 0733 2583246; fax: +39 0733 2583205. Email: roy.cerqueti@unimc.it} \\
$^{\natural}$  {\small Catholic University of Milan, Department of Mathematics, Finance and Econometrics} \\ {\small Email: gianpaolo.clemente@unicatt.it} \\
$^{\sharp}$  {\small University of Milano-Bicocca, Department of Statistics and Quantitative Methods} \\ {\small Email: rosanna.grassi@unimib.it}
}
\maketitle

\begin{abstract}

This paper presents a new definition of the community structure of a
network, which takes also into account how communities are
stratified. In particular, we extend the standard concept of
clustering coefficient and provide the local $l$-adjacency
clustering coefficient of a node $i$. We define it as an opportunely
weighted mean of the clustering coefficients of nodes which are at
distance $l$ from $i$. The stratus of the community associated to
node $i$ is identified by the distance $l$ from $i$, so that the
standard clustering coefficient is a peculiar local $l$-adjacency
clustering coefficient at stratus $l=0$. As the distance $l$ varies,
the local $l$-adjacency clustering coefficient is then used to infer
insights on the community structure of the entire network. Empirical
experiments on special business networks are carried out. In
particular, the analysis of air traffic networks validate the
theoretical proposal and provide supporting arguments on its
usefulness.

\end{abstract}


\section{Introduction}

Business researchers and economists are paying a growing attention
to the relevant theme of communities detection and to the related
complex networks models, as recent literature suggests. In the area
of financial economics, \cite{ausloos2007clusters} deals with the
clusters of networks of economies by using data on the Gross
Domestic Product of a group of countries. Under an international
economics perspective, \cite{Barigozzi2011,Piccardi2012} discuss the
assessment of the communities in the context of the world trade
network. The theme of the brand communities in the framework of
entrepreneurship and social science is also largely debated (see
e.g. \cite{Essamri2019, Lin2017, Zaglia2013}). It is worth
mentioning the role of communities in corporate organizations and
firms cooperation (see e.g. \cite{Vitali2014, Wilkinson2002}). In
general, a managerial architecture is grounded on the relationships
among agents, and network structures are particularly effective in
offering a representation of the complexity of the relationship
among the involved actors. Firms and managers are usually
interconnected through an intricate weave. Hence, discovering how
they are clustered at every level of the structure can give useful
insights in entrepreneurial strategic choices (\cite{Haakansson2002,
Wilkinson2002, Zaglia2013}).

The community structure around a node is generally represented and
measured through its clustering coefficient. Such a coefficient can
be defined in an operative way, being computed as the relative
number of actual triangles to which the vertex belongs over the
hypothetical ones. It has been developed in all the cases of
weighted, unweighted, directed and undirected networks. In this
respect, we mention the classical contributions of
\cite{Barrat_2004, Fagiolo_2007, Onnela_2005, Watts_1998} and the
non recent but highly informative monograph \cite{WasFaust}.
Recently, \cite{CleGra} contains a relevant extension of the
clustering coefficient proposed by \cite{Fagiolo_2007};
\cite{Rotundo_2010} discusses the clustering coefficient in presence
of already established communities for directed networks; \cite{Roy}
presents a concept of clustering coefficient which also includes the
presence of missing indirect links in the construction of the
triangles. The association between communities and clustering
coefficients is quite natural. Triangles are the easiest geometric
visualizations of the communities, providing a picture of
non-exclusive interactions among different agents. Clustering
coefficient could be used then as an indicator of the node position
with respect to communities.
However, even if it describes the linkages around a vertex, this
coefficient is a local measure hence it does not capture the
topology of the whole network. Moreover, the global clustering
coefficient of the network -- which is simply obtained by taking the
average of the local versions over all the nodes of the network --
often is not very informative. Indeed, a network can be highly
clustered at a local level but not at a global level, and this
suggests that the average of local clustering could not well
represent the global characteristics of the network. This drawback is peculiarly relevant, since
it is often crucial to assess the position of a node with respect to
communities in which the node itself is not directly involved (see
\cite{Estrada2011}).

This paper extends the concept of community structure of a node $i$
by considering the role and the relevance of the communities which
are placed at a geodesic distance $l\geq 1$ from the considered node
$i$. In details, we define the local $l$-adjacency clustering
coefficient of a node $i$ as an opportunely weighted mean of the
clustering coefficients of the nodes at geodesic distance $l$ from
$i$. In this way, we present a model of stratified communities,
where the stratification is ruled by the varying geodesic distances
from the reference node. In particular, the stratified community at
stratus $l$ is defined through the clustering coefficients of the
nodes at geodesic distance $l$ from $i$. In this respect, the
geodesic distance from the node represents the stratus of the
community.

Assessing the stratified community structure of a network represents
a crucial point for understanding the contextualization of the nodes
within the overall system. For instance, assume that a node $i$ has a low clustering coefficient but the nodes at a
distance $l>1$ from it exhibit a strong community level. In this
case, the node $i$ is not embedded in a powerful community
but it is surrounded by a high level of mutual interconnections for
the nodes quite far from it. This situation can be interpreted
under the perspective of shocks propagation. Indeed, if a node, located at a
geodesic distance from $i$ greater than $l$, receives a shock, such a shock
could hardly reach $i$, since $i$ is surrounded by highly clustered
nodes at a large distance from it and by an empty space -- in terms
of connections -- around it. In other words, nodes at distance $l$ might be
viewed as a sort of barrier for $i$, absorbing the external
solicitations at larger distances. This stops the propagation to $i$,
because of the absence of a community structure close to the node.

To test the proposed measure we perform an empirical application to
the paradigmatic business network related to the U.S. domestic air
traffic. Findings confirm the effectiveness of these measures in
seizing the peculiar characteristics of different hubs in the
airport network. We observe that larger hubs are not only highly
clustered but also connected to strong communities. When higher
strati are analysed, the effects of indirect connections with remote
airports are emphasized. Referring to large and medium hubs, the
proposed approach allows to emphasize the strategic role of such
nodes in the airport network system. Finally, although the
considered empirical network is only weakly asymmetric and it is
typically analysed as an undirected network in the literature
(see, for instance, \cite{Barrat_2004, Colizza, CleGra, Jia, Jia2014, Opsahl}), we show that a separate evaluation of directed
paths at high levels can be useful to identify specific patterns in
terms of in- and out-communities.

The rest of the paper is organized as follows. In Section
\ref{Prel}, we present the mathematical preliminaries and the basic
notation used in the article. In Section \ref{L_CC}, the concepts of
local $l$-adjacency clustering coefficient -- being $l$ the stratus
of the related community -- is defined for both undirected and
directed case. In Section \ref{NA} the new conceptualization of
clustering coefficients is applied to the U.S. airport network, in
order to discuss its stratified community structure. Conclusions are
in Section \ref{Concl}.

\section{Preliminary definitions and notations}\label{Prel}

We briefly present the mathematical definitions used in the paper.
A graph $G=(V,E)$ is identified by a set $V$ of $n$ vertices and a set $E$ of $m$ unordered pairs of vertices
(called edges). Vertices $i$ and $j$ are said to be adjacent if $(i,j)\in E$. Graphs considered here will be always without loops. The degree $d_{i}$ of the vertex $i$ is the number of its adjacent nodes.
A path between two vertices $i$ and $j$ is a sequence of distinct vertices and edges between $i$ and $j$. In this case, $i$ and $j$ are connected. $G$ is connected if every pair of vertices is connected. The distance $ d\left( i,j\right) $ is the length of any shortest path (or geodesic) between $i$ and $j$. If $i$ and $j$ are not connected, then $ d\left( i,j\right)=\infty $.
The diameter $diam(G)$ of a connected graph is the length of any longest geodesic.

Given a connected graph $G$, we define the set: $$N_{i}(l)=\{j
\in V | d(i,j)=l\}$$ of nodes $j \in V$ which are at distance $l$
with respect to the node $i$, where $l=1,\dots,diam(G)$. The
cardinality $|N_{i}(l)|$ is denoted as $d_{i}(l)$.
Notice that $N_{i}(1)$ is the set of nodes adjacent to
$i$, so that $d_{i}(1)=d_{i}$.

A graph $G$ is weighted when a positive real number
$w_{ij}>0$ is associated with the edge $(i,j)$. $w_{ij}=0$ if nodes
$i$ and $j$ are not adjacent.
In particular, when $w_{ij}=1$ if $(i,j) \in E$, then the
graph is unweighted. Thus, the unweighted case can be viewed
as a particular weighted one. For this reason, we use in this paper
only the general concept of weighted graphs and we denote a weighted
graph with its weights simply as weighted network.

The strength of the vertex $i$ is defined as
$s_{i}=\sum_{j=1}^{n}w_{ij}$.

In general, both adjacency relationships between vertices of $G$ and
weights on the edges are described by a nonnegative, real $n$-square
matrix $\mathbf{W}$ (the \emph{weighted adjacency matrix}), with
entries $w_{ij}$.

A weight can be also associated to every geodesic of a connected graph $G$
in the following way: let $\mathcal{G}_{ij}$ be the set of the
geodesics connecting the vertices $i$ and $j$; the generic element
of $\mathcal{G}_{ij}$ is $g=g_{ij}$.

We observe that more than one geodesic connecting $i$ and $j$ could
exist, so that, in general, $|\mathcal{G}_{ij}|\geq 1$. Recalling that
all the geodesics in $\mathcal{G}_{ij}$ have the same length, a
unique integer $l=l(i,j)$ exists such that the length of all the
paths in $\mathcal{G}_{ij}$ is $l$. The role of $l$ is crucial for
the following arguments. Therefore, $l$ will be explicitly added in
the notation when needed so that, for instance, $\mathcal{G}_{ij}$
will be $\mathcal{G}_{ij}(l)$, $g$ will become $g(l)$ and so on.

The weight of $g_{ij}(l)$ is the sum of the weights of its edges and
we denote it with $w_{ij}(l,g)$. This allows us to define the $l$-th
order strength of the node $i$ in this setting as
$$s_{i}(l)=\sum_{j \in N_{i}(l)} w_{ij}(l),$$
where $w_{ij}(l)=\min\limits_{g \in
\mathcal{G}_{ij}(l)}\{w_{ij}(l,g)\}.$  Notice that when $j \in
N_{i}(1)$, then $\mathcal{G}_{ij}(1)=\{(i,j)\}$; hence
$w_{ij}(1)=w_{ij}$ and $s_{i}(1)$ is the strength $s_{i}$ of the
vertex $i$.


A directed graph $D=(V,E)$ is obtained from $G$ by adding to
its edges a direction and $G$ is the underlying graph of $D$. In
this case, the links between couples of nodes are called
directed edges or arcs.
In a weighted directed graph, a weight $w_{ij}>0$ is associated with
the directed edge $(i,j)$ and, in general, the matrix $\mathbf{W}$
is not symmetric. In fact, since bidirectional edges between a pair
of nodes can exist, both $w_{ij}$ and $w_{ji}$ can be positive with
$w_{ij} \neq w_{ji}$.

A directed path from $i$ to $j$ is a sequence of distinct
vertices and arcs from $i$ to $j$ such that every arc has the same
direction; in this case, we say that $j$ is reachable from
$i$ and we call this \emph{out-path} of the node $i$. The distance
$\overrightarrow{d}(i,j)$ from $i$ to $j$ is the length of
such shortest out-path (or \emph{out-geodesic}) if any, otherwise $\overrightarrow{d}(i,j)=\infty$.

Since directed paths from $j$ to $i$ can also exist, we define the
\emph{in-path} of the node $i$ as the directed path from $j$ to $i$
and we denote with $\overleftarrow{d}(i,j)$ the length of any
shortest such in-path (or \emph{in-geodesic}). If $i$ is not reachable from $j$, then $\overleftarrow{d}(i,j)=\infty$.

If $i$ and $j$ are mutually reachable, both in and
out-geodesics of $i$ exist, although the distances
$\overleftarrow{d}(i,j)$ and $\overrightarrow{d}(i,j)$ can be
different. $D$ is strongly connected if every two vertices
are mutually reachable.

$D$ is weakly connected if the underlying graph $G$ is
connected. That it means that a geodesic $g$
between $i$ and $j$ exists in the underlying graph $G$. In this case, distance $d(i,j)$ is finite for all $i,j \in V$.

In addition, we define the following sets:

\begin{enumerate}
\item $\overrightarrow{N}_{i}(l)=\{j \in V | \overrightarrow{d}(i,j)=l\}$, for each $l=1,\dots,diam(G)$;
\item $\overleftarrow{N}_{i}(l)=\{j \in V | \overleftarrow{d}(i,j)=l\}$, for each $l=1,\dots,diam(G)$.
\end{enumerate}

Moreover, according to what we did above, we can define
$\overrightarrow{\mathcal{G}}_{ij}$ and
$\overleftarrow{\mathcal{G}}_{ij}$ as the sets of the out-geodesics
and in-geodesics connecting the vertices $i$ and $j$, respectively;
the generic elements of $\overrightarrow{\mathcal{G}}_{ij}$ and
$\overleftarrow{\mathcal{G}}_{ij}$ are
$\overrightarrow{g}=\overrightarrow{g}_{ij}$ and
$\overleftarrow{g}=\overleftarrow{g}_{ij}$, respectively.

Hence, the definition of weighted in- and out-geodesics can be
easily given by setting:
\begin{enumerate}
\item $\overrightarrow{w}_{ij}(l)=\min\limits_{\overrightarrow{g} \in
\overrightarrow{\mathcal{G}}_{ij}(l)}
\{\overrightarrow{w}_{ij}(l,\overrightarrow{g})\}$;
\item $\overleftarrow{w}_{ij}(l)=\min\limits_{\overleftarrow{g} \in \overleftarrow{\mathcal{G}}_{ij}(l)}\{\overleftarrow{w}_{ij}(l,\overleftarrow{g})\}$.
\end{enumerate}

This allows us to define the $l$-th order in and out-strength of the node $i$ as \\
$$\overrightarrow{s}_{i}(l)=\sum_{j \in \overrightarrow{N}_{i}(l)} \overrightarrow{w}_{ij}(l),$$
$$\overleftarrow{s}_{i}(l)=\sum_{j \in \overleftarrow{N}_{i}(l)} \overleftarrow{w}_{ij}(l).$$

\section{Stratified communities}
\label{L_CC}

The aim of this section is to define a new indicator of community
structure around a node $i$ based on the mutual interconnections
between nodes at different distances from $i$. This new indicator
uses an extended idea of clustering coefficient, moving along
shortest paths. Hence, we are providing a conceptualization of
the stratified communities around the nodes of a network.

We introduce the indicator discussing separately the undirected and
directed case.

\subsection{Local $l$-adjacency clustering coefficient: undirected case}

Let $\mathbf{P}(l)=[p_{ij}(l)]_{i,j \in V}$ for $l=1,\dots, diam(G)$
be the matrix such that:
\begin{equation}
p_{ij}(l)= \begin{cases} \frac{w_{ij}(l)}{s_{i}(l)} & \mbox{if }  j
\in N_{i}(l) \mbox{ and } N_{i}(l) \neq \emptyset, \\ 0 &
\mbox{otherwise}.
\end{cases}
\label{pW}
\end{equation}
For the sake of completeness, the definition of the matrix
$\mathbf{P}(l)$ can be extended to the case $l=0$, by setting
$\mathbf{P}(0)=\mathbf{I}$, being $\mathbf{I}$ the identity matrix.

We define the vector of the \emph{local $l$-adjacency clustering
coefficients} of the nodes of the network 
$\mathbf{c}(l)=[c_i(l)]_{i \in V}$, as:
\begin{equation}
\mathbf{c}(l)=\mathbf{P}(l) \mathbf{c}, \label{Wclustering}
\end{equation}
where $\mathbf{c}=[c_i]_{i \in V}$ is the vector whose element
$c_{i}$ is the weighted clustering coefficient of node $i$,
defined in Barrat et al. (\cite{Barrat_2004}). \\

Notice that, when $l=0$, formula (\ref{Wclustering}) gives
$\mathbf{c}(0)=\mathbf{c}$, and then we recover the weighted local
clustering coefficient defined in \cite{Barrat_2004}.

When $l=1$, the local $l$-adjacency clustering coefficient
$\mathbf{c}(1)=[c_i(1)]_{i \in V}$ is the vector of elements:
\begin{equation}
\label{ci} c_{i}(1)=\frac{1}{s_i}\sum_{j \in V} w_{ij}c_j,
\end{equation}
where each element represents the weighted average of the clustering coefficients $c_{j}$ of the nodes $j$ which are adjacent to $i$.
This is true also for $l>1$, as in general, formula (\ref{Wclustering}) states that $c_{i}(l)$ is
the weighted average, with weights $w_{ij}(l)$, of clustering coefficients $c_{j}$ of nodes $j$
which are at distance $l$ with respect to the node $i$.
Furthermore, it is noteworthy that, in case of an unweighted graph,
the coefficient $c_{i}(l)$ reduces to a classic arithmetic mean.

The elements of the vector defined in (\ref{Wclustering}) give
insights about the community structure at a specific distance from
nodes of the graph. Indeed, the element $l$ associated to the
proposed definition of the clustering coefficient explains how the
nodes at distance $l$ from $i$ form a community in the graph.

Large values of such clustering coefficients of $i$ at high levels
$l$ suggest that $i$ is connected with well-established communities
which are far from the node itself.
In other words, the analysis of $\mathbf{c}(l)$ with $l=0,1, \dots,
diam(G)$ leads to a complete view of the graphs in terms of spatial
communities, and this might give insights on how shocks propagate.
The quantity $c_{i}(l)$ describes the \emph{community structure at
stratus $l$ around the node $i$} and the set of vectors $\mathbf{c}(l)$
defines the \textit{stratified community structure of the network} as
the value of $l$ varies.


In order to measure the overall community structure around a node
$i$, we then introduce the vector $\mathbf{h}=[h_i]_{i=0,1,\dots,
n}$ such that
\begin{equation}
h_i=\sum_{l=0}^{diam(G)} x_{l}c_{i}(l), \label{h:clustering}
\end{equation}
where $x_l \geq 0$, for each $l$, and $\sum_{l=0}^{diam(G)}
x_{l}=1$.

The vector $\mathbf{h}=[h_i]_{i=0,1,\dots, n}$ allows tracking the
whole community structure around a single node $i$, and takes into
account all the strati. Notice that the selection of a peculiar
distribution of the weights $(x_0, x_1, \dots, x_{diam(G)})$
provides the meaning of the concept of stratified community for all
the nodes of the graph. In particular, high polarization of such a
distribution at low (high) level $l$ leads to core-based
(periphery-based) identification of the community structure. The
special case $x_l = 1$ focuses attention only to communities at
stratus $l$.





\subsection{Local $l$-adjacency clustering
coefficients: directed case}\label{sec:DWcase}

We consider a directed, weighted and weakly connected graph $D$.\\
As already pointed out in Section \ref{Prel}, in addition to
weighted paths, also weighted in- and out-paths can exist and we can
focus only on a specific pattern (out-path or in-path for the node
$i$), or we can consider all edges' directions. Each choice is
reasonable and depends on the kind of problem we deal with.

For $l=1, \dots, diam(G)$, we define the matrix
$\mathbf{\bar{P}}(l)$ with the following entries:
\begin{equation}
\bar{p}_{ij}(l)= \begin{cases} \frac{\bar{w}_{ij}(l)}{\bar{s}_{i}(l)} &
\mbox{if }  j \in \bar{N}_{i}(l) \mbox{ and } \bar{N}_{i}(l) \neq \emptyset,\\
0 & \mbox{otherwise,}
\end{cases}
\label{eq:matrixPdirected}
\end{equation}
where:
\begin{itemize}
    \item[$(a)$] $\bar{N}_{i}(l)=\overrightarrow{N}_{i}(l)$, $\bar{w}_{i,j}(l)=\overrightarrow{w}_{ij}(l)$ and $\bar{s}_{i}(l)=\overrightarrow{s}_{i}(l)$ if only out-paths of node $i$ are considered. In this case, $\mathbf{\bar{P}}(l)$ will be denoted by $\overrightarrow{\mathbf{P}}(l)$;
    \item[$(b)$]  $\bar{N}_{i}(l)=\overleftarrow{N}_{i}(l)$, $\bar{w}_{i,j}(l)=\overleftarrow{w}_{ij}(l)$ and $\bar{s}_{i}(l)=\overleftarrow{s}_{i}(l)$ if only in-paths of node $i$ are considered. In this case, $\mathbf{\bar{P}}(l)$ will be denoted by $\overleftarrow{\mathbf{P}}(l)$;
    \item[$(c)$]  $\bar{N}_{i}(l)=N_{i}(l)$, $\bar{w}_{i,j}(l)=w_{ij}(l)$ and $\bar{s}_{i}(l)=s_{i}(l)$ if all the directions of the edges are taken into account. In this case,  $\mathbf{\bar{P}}(l)$ will be denoted by $\mathbf{P}(l)$.
\end{itemize}
Also in this case, we set $\mathbf{\bar{P}}(0)=\mathbf{I}$.


The definition of local $l$-adjacency clustering coefficients,
introduced in formula (\ref{Wclustering}) has to be extended to the
three cases $(a)$, $(b)$ and $(c)$.


Indeed, in the specific context of directed graphs, edges pointing
in different directions have a completely different interpretation
in terms of the resulting flow pattern.  To this aim, alternative
in-type or out-type local $l$-adjacency clustering coefficients can
be also obtained as:
\begin{equation}
\textbf{c}^{in}(l)=\overleftarrow{\mathbf{P}}(l) \textbf{c}^{in}.
\label{WDclusteringin}
\end{equation}
\begin{equation}
\textbf{c}^{out}(l)=\overrightarrow{\mathbf{P}}(l) \textbf{c}^{out}.
\label{WDclusteringout}
\end{equation}
where $\overleftarrow{\mathbf{P}}(l)$ and
$\overrightarrow{\mathbf{P}}(l)$ are matrices defined in formula
(\ref{eq:matrixPdirected}) whose elements are computed considering
only in-paths (case $(b)$) or out-paths (case $(a)$), respectively.
$\textbf{c}^{in}=[c^{in}_i]_{i \in V}$ and
$\textbf{c}^{out}=[c^{out}_i]_{i \in V}$ are the vectors whose
elements $c^{in}_{i}$ and $c^{out}_{i}$ are the in and out local
local clustering coefficients defined in \cite{CleGra}. These two
coefficients convey information about clustering of two different
patterns (in or out) within tightly connected directed
neighbourhoods.

According to the case $(c)$, we define the local $l$-adjacency
clustering coefficients as:
\begin{equation}
\textbf{c}^{all}(l)=\mathbf{P}(l) \textbf{c}^{all},
\label{WDclusteringall}
\end{equation}
where $\textbf{c}^{all}=[c^{all}_i]_{i \in V}$ is the vector whose
elements $c^{all}_{i}$ are the weighted and directed clustering coefficients provided in \cite{CleGra}.

The difference between the vectors of local $l$-adjacency clustering
coefficients in (\ref{Wclustering}) and (\ref{WDclusteringall}) lies
in the considered definition of triangles. The same triple of nodes
might be associated to one triangle (no directions of the arcs to be
taken into account) in the former case and two of them (two possible
directions for the arcs) in the latter one. In the special case of
absence of bilateral arcs, the following result holds true:
\begin{proposition}
\label{prop} Let D be a directed graph. If the graph $D$ has not
bilateral arcs, then $\mathbf{c}^{all}(l)=\frac{\mathbf{c}(l)}{2}$,
where $\mathbf{c}(l)$ is the vector of the local $l$-adjacency
clustering coefficients of the undirected underlying graph $G$.
\end{proposition}

\begin{proof}
We define $\textbf{W}(D)$ and $\textbf{W}(G)$  the weighted
adjacency matrices of the graphs $D$ and $G$. The local clustering
coefficient for weighted and directed network is defined as (see
\cite{CleGra}):

\begin{equation}
    c_{i}^{all}=\frac{\frac{1}{2}\left[ \left( \textbf{W}(D)+\textbf{W}^{T}(D)\right) \left(
    \textbf{A}(D)+\textbf{A}^{T}(D)\right) ^{2}\right] _{ii}}{s_{i}(1,D)\left( d_{i}(1,D)-1\right)
    -2s_{i}^{\leftrightarrow }(1,D)}
\label{call}
\end{equation}
where $s_{i}(1,D)$ and $d_{i}(1,D)$ are respectively the total strength and the degree of the node $i$, whereas $s_{i}^{\leftrightarrow}(1,D)$ is the
strength related to bilateral arcs between $i$ and its
adjacent nodes.

Since the graph $D$ has not bilateral arcs,
$s_{i}^{\leftrightarrow}(1,D)=0$. Additionally, if nodes $i$ and $j$
are adjacent, a (unique) weighted arc between $i$ and $j$ exists so
that if $\overrightarrow{w}_{ij}(1,D)> 0$ then
$\overleftarrow{w}_{ij}(1,D)=0$ or vice versa. As a consequence,
$\textbf{W}(D)+\textbf{W}^{T}(D)=\textbf{W}(G)$ and
$\textbf{A}(D)+\textbf{A}^{T}(D)=\textbf{A}(G)$.

Furthermore,
$$s_{i}(1,D)=\overrightarrow{s}_{i}(1,D)+\overleftarrow{s}_{i}(1,D)=\sum_{j
\in \overrightarrow{N}_{i}(1)} \overrightarrow{w}_{ij}(1,D)+\sum_{j
\in \overleftarrow{N}_{i}(1)} \overleftarrow{w}_{ji}(1,D)=\sum_{j
\in N_{i}(1)}w_{i,j}(1,G)=s_{i}(1,G).$$
A similar chain of equalities
entails that  $d_{i}(1,D)=d_{i}(1,G)$.

Hence, (\ref{call}) yields, $\forall i=1,...,n$:

\begin{equation}
c_{i}^{all}=\frac{\frac{1}{2}\left[ \textbf{W}(G)\textbf{A}^{2}(G)\right] _{ii}}{s_{i}(1,G)\left( d_{i}(1,G)-1\right) }   =\frac{c_{i}}{2}
\label{call1}
\end{equation}

and, by formula (\ref{WDclusteringall}):
\begin{equation}
\textbf{c}^{all}(l)=\mathbf{P}(l)
\textbf{c}^{all}=\frac{1}{2}\mathbf{P}(l)
\textbf{c}=\frac{1}{2}\textbf{c}(l).
\end{equation}

    \end{proof}

Analogously to formula (\ref{h:clustering}), we can define
$\textbf{h}^{in}=[h^{in}_i]_{i \in V}$,
$\textbf{h}^{out}=[h^{out}_i]_{i \in V}$, and
$\textbf{h}^{all}=[h^{all}_i]_{i \in V}$, where:
\begin{equation}
h^{in}_i=\sum_{l=0}^{diam(G)} x_l {c}_i^{in}(l), \qquad
h^{out}_i=\sum_{l=0}^{diam(G)} x_l {c}_i^{out}(l), \qquad
h^{all}_i=\sum_{l=0}^{diam(G)} x_l {c}_i^{all}(l).\label{hifrecce}
\end{equation}

\subsection{How the local $l$-adjacency clustering
coefficients work: an illustrative example} \label{sec:example} The
classical clustering coefficient does not give insights on both the
topological structure and the stratified community structure of
the whole network, being a measure of the local community structure
concentrated around the nodes of the network. For the same reason
also the global clustering coefficient of the network, which is
given by the average of the local version around the nodes, is often
not very informative. Furthermore, a network can be highly clustered
at a local level but not on a global level, so that the average of
the local clustering could not well represent the global
characteristics of the network (see \cite{Estrada2011}).

To illustrate how the local $l$-adjacency clustering coefficients
effectively map the communities structure in the network, we show
how they work by comparing two small graphs. For the sake of
simplicity, we limit our investigation to undirected and unweighted
graphs, as interesting remarks can be done also in this very
simplified, but meaningful, case.




Let us consider the two graphs $G$ and $G'$, sharing the same number
of nodes (namely, 9) and same diameter (namely, 4), but different
topology, having $G'$ more arcs than $G$ (hence, showing a stronger
community structure, see Figure \ref{F:G}). In both graphs, node 1
shares the same neighbours and, being part of the same triangles, it
has the same clustering coefficient $c_{1}(0)$. This coefficient is
then not effective in capturing the actual position of the node with
respect to the communities in the network. Indeed, the classical
clustering coefficient gives insights about how the node 1 is
embedded in a cohesive group only respect to its neighbours.

To have a more reliable information about how the node is located in
respect to the whole structure, in particular to existing
communities, we need to analyse its position looking deeper than its
neighbours and the $l$-adjacency clustering coefficients (with
$l=1,..,4$) in this sense are meaningful.


\begin{figure}[!h]
    \centering
    \includegraphics[scale=0.3]{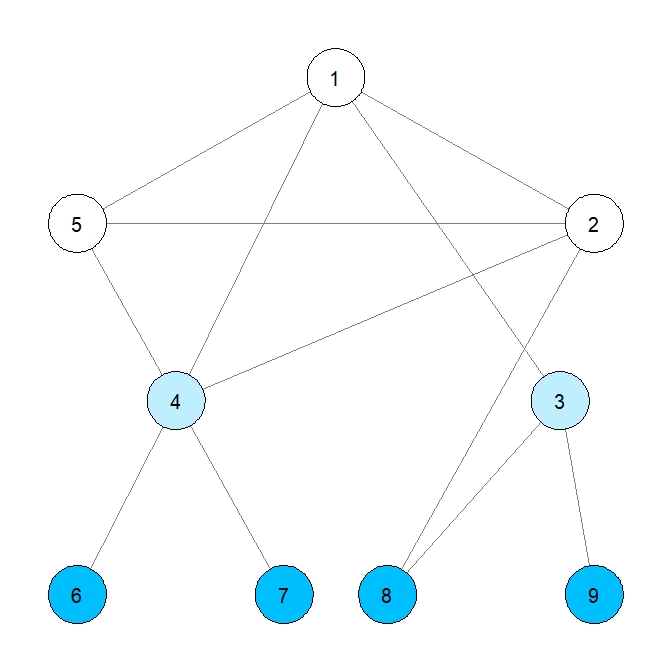}
    \includegraphics[scale=0.3]{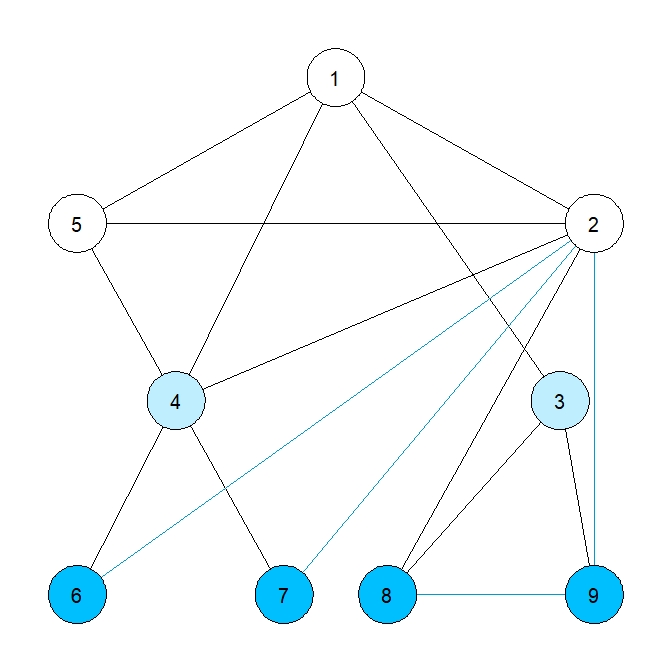}
    \caption{Graphs $G$ (left side) and $G'$ (right side).}
    \label{F:G}
\end{figure}

\begin{table}[!h]
    \tiny
    \centering
    \begin{tabular}{| c | c | c|  c| c| c| c| c| c| c| }
        \hline
        \hline
        Node & $c_{i}(0)$ & $c_{i}(1)$  & $c_{i}(2)$ & $c_{i}(3)$    & $c_{i}(4)$ & $h_{i}$  &  $h_{i}$  & $h_{i}$   \\
        & & & &  & & Decreasing Weights &  Uniform Weights & Increasing Weights \\

        \hline
        \hline

        1 &0.50 &0.45   &0  &0  &0  &0.32   &0.19   &0.09 \\
        2 &0.50 &0.45   &0  &0  &0  &0.32   &0.19   &0.09 \\
        3 &0    &0.17   &0.60   &0  &0  &0.12   &0.15   &0.11 \\
        4 &0.30 &0.40   &0  &0  &0  &0.22   &0.14   &0.07 \\
        5 &1    &0.43   &0  &0  &0  &0.53   &0.29   &0.14 \\
        6 &0    &0.30   &0.50   &0  &0  &0.14   &0.16   &0.11 \\
        7 &0    &0.30   &0.50   &0  &0  &0.14   &0.16   &0.11 \\
        8 &0    &0.25   &0.45   &0  &0  &0.12   &0.14   &0.09 \\
        9 &0    &0      &0.25   &0.60   &0  &0.10   &0.17   &0.17 \\
        \hline
        \hline
    \end{tabular}
    \caption{Graph $G$: $l$-adjacency clustering coefficients, and elements of vector $\mathbf{h}$ with different choices of weights.}
    \label{tab:CC4a}
\end{table}

\begin{table}[!h]
    \tiny
    \centering
    \begin{tabular}{| c | c | c|  c| c| c| c| c| c| c| }
        \hline
        \hline
        Node & $c_{i}(0)$ & $c_{i}(1)$  & $c_{i}(2)$ & $c_{i}(3)$    & $h_{i}$  &  $h_{i}$  & $h_{i}$   \\
        & & & &  & Decreasing Weights &  Uniform Weights & Increasing Weights \\

        \hline
        \hline

        1 &0.50      &0.53  & 0.83  &0   &0.50 &     0.47 &  0.34 \\
        2 &0.29      &0.76  & 0.33  &0   &0.37   &0.35   &0.24 \\
        3 &0.33      &0.61  & 0.60   &1 &    0.52   & 0.63 &     0.76 \\
        4 &0.50      &0.76  & 0.56  &0   &0.51   &0.45  & 0.31 \\
        5 &1     &0.43   &0.73  &0   &0.70   & 0.54     & 0.36 \\
        6 &1     &0.39   &0.77   &0.33   &0.74 &     0.62    &0.53 \\
        7 &1     &0.39   &0.77   &0.33   &0.74  & 0.62   &0.53 \\
        8 &0.67     & 0.43   &0.80  &0  & 0.55 &     0.47 &  0.34 \\
        9 &0.67     & 0.43   &0.80  &0  & 0.55 &     0.47 &  0.34 \\

        \hline
        \hline
    \end{tabular}
    \caption{Graph $G'$: $l$-adjacency clustering coefficients, and elements of vector $\mathbf{h}$ with different choices of weights.}
    \label{tab:CC4b}
\end{table}

Comparing the values of the clustering coefficients in Tables
\ref{tab:CC4a} and \ref{tab:CC4b}, the coefficient of node 1
decreases with respect to $l$ in the case of $G$ (from 0.5 to 0)
moving through paths of length greater than 1, whereas in case $G'$
it increases till $l=2$ (from 0.5 to 0.83) and becomes null for
$l=3$.

Therefore, the different strati of the community structure around
the node 1 well reflect the position of such a node with respect to
the way the other nodes are connected in the structure.

Through elements $h_i$ we are able to simultaneously consider all
the communities at the different strati. We control the impact of
each coefficient through their weights $x_l$. We consider here three
possible scenarios for the weights $x$'s:
\begin{itemize}
\item Decreasing weights: $x_{l}=\frac{(l+1)^{-1}}{\sum_{h=0}^{diam(G)}(h+1)^{-1}}=\frac{(l+1)^{-1}}{H_{G}}$ where $H_{G}$ is the harmonic number of order $diam(G)+1$, for each $l$;
\item Uniform weights: $x_{l}=\frac{1}{diam(G)+1}$, for each $l$;
\item Increasing weights: $x_{l}=\frac{(l+1)}{\sum_{h=0}^{diam(G)}(l+1)}=\frac{2(l+1)}{\left(diam(G)+1\right)\left(diam(G)+2\right)}$, for each
$l$.
\end{itemize}

Intuitive interpretations of the weights arise. Decreasing weights,
for instance, reduce the impact on the node of high distances when
assessing the community. Notice that, the elements in $\mathbf{h}$
do not provide similar information of the classical average
clustering coefficient. Instead, we are measuring the position of
the node inside the network looking at each stratus. These
indicators then provide an overall look and, at the same time, they
track the node distances from the communities in the network.

\section{Empirical experiments}\label{NA}
In order to see how the proposed indicator is effective in
describing the stratified communities of a node, we test it on the
peculiar business network of the U.S. airport, where nodes are the
airports and arcs are weighted on the basis of the flights scheduled
among them in a given year. The considered reference year in the
proposed experiments is 2017.
The network is constructed by using the Air Carrier Statistics
database (available on the U.S. Department of Transportation\footnote{Data are collected by the
    Office of Airline Information, Bureau of Transportation Statistics,
    Research and Innovative Technology Administration.}), also
known as the T-100 data bank, that contains domestic and
international airline market and segment data. Both certificated U.S. air carriers and foreign
carriers (having at least one point of service in the United States
or one of its territories) report monthly traffic
information.
The weight of an arc corresponds to the number of emplaned
passengers\footnote{The term ``emplaned passengers'', widely used in
the aviation industry, refers to passengers boarding a plane at a
particular airport. Since the majority of airport revenues are
generated, directly or indirectly, by emplaned passengers, this
number is the most important air traffic metric. Data consider the
total number of revenue passengers boarding an aircraft (including
originating, stopover, and transfer passengers) in both scheduled
and non-scheduled services.}. It considers revenue emplaned
passengers within the U.S., and passengers emplaned outside U.S. but
deplaned within the U.S. as well.

In the reference year 2017, the airport network has 1701 nodes and
27005 arcs, considering both domestic and international flights.
Density is around 0.001, showing a very sparse network. Moreover,
significant differences are observed between big and small airports.
To give a preliminary idea of the network, Figure \ref{F:BI2Ynet}
depicts the U.S. domestic airport network. In order to preserve the
clarity of the figure, we reported only arcs with weights greater
than 95th percentile (equal to 198,540) of weights' distribution. In
other words, for the sake of simplicity we are displaying only
routes with more than about 200,000 enplaned passengers.

\begin{figure}[!h]
    \centering
    \includegraphics[scale=0.8]{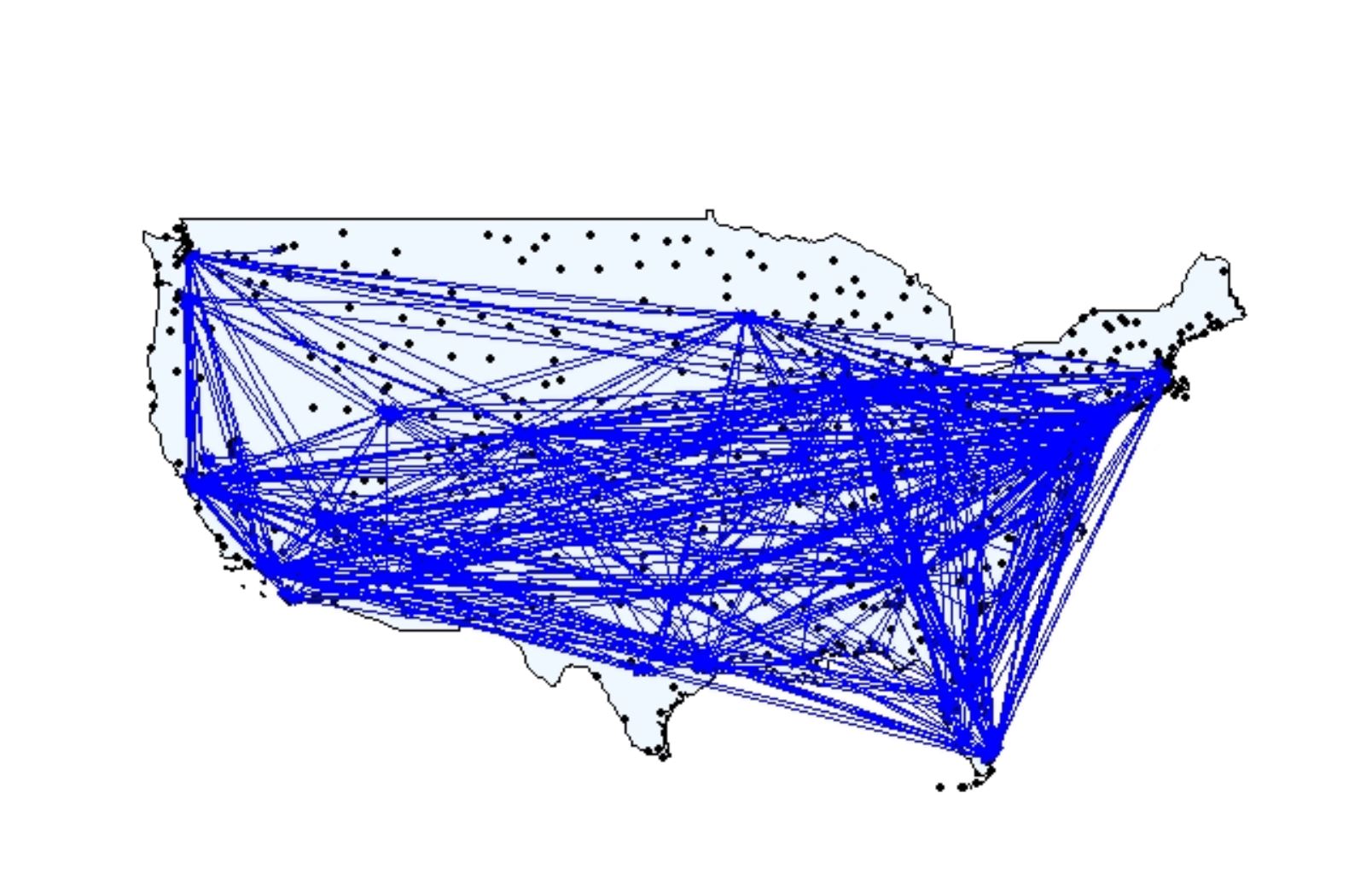}
\caption{Domestic U.S. airport network built on the basis of 2017
data. Only arcs with weights greater than 95th percentile of weights'
distribution have been reported.}
    \label{F:BI2Ynet}
\end{figure}

Figure \ref{F:DSdist} reports the distributions of total strength
for U.S. airports, capturing the total passenger traffic during
2017. Airports have been split according to Federal Aviation
Administration (FAA) categories. According to FAA, a large hub is an
airport which accounts for at least 1\% of total U.S. passenger
enplanements. A medium hub is defined as an airport accounting for a
percentage of the total passenger enplanements ranging between
0.25\% and 1\% (see Tables \ref{tab:CC3} and \ref{tab:CC3b} in the
Appendix for the list of large and medium U.S. hubs). A small hub is
associated to a percentage ranging between 0.05\% and 0.25\%. Last
categories concern smaller airports, that are divided between
non-hub and non-primary if they have respectively more or less than
10,000 annual passengers.

\begin{figure}[!h]
    \centering
    \includegraphics[scale=0.3]{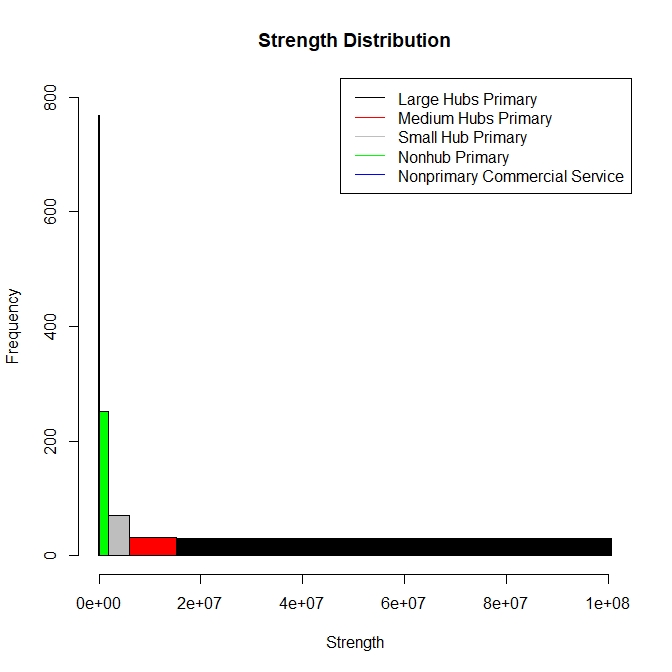}
    \includegraphics[scale=0.3]{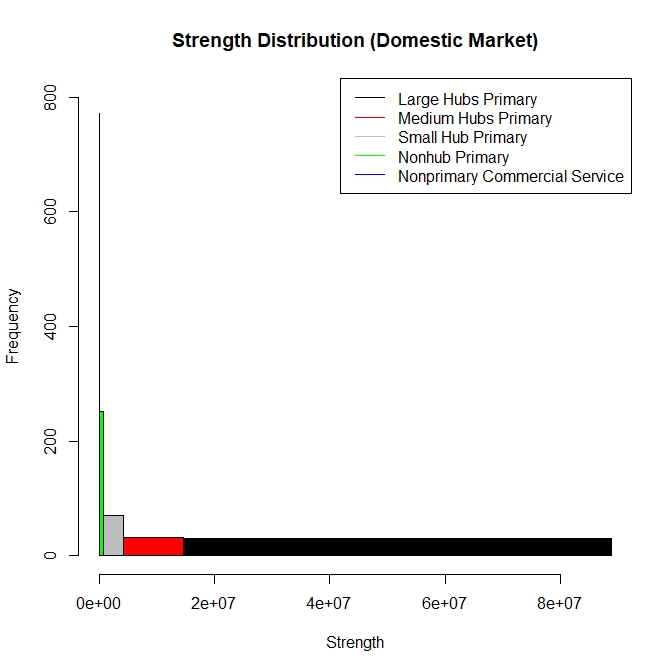}
\caption{Total strength distribution of U.S. airports for domestic
and international markets and only domestic market respectively.
Airports are classified according to FAA categories.}
    \label{F:DSdist}
\end{figure}

System-wide passenger enplanements is not far from 900 millions of
passengers. The 30 large hubs move 70\% of the passengers, and this
ratio becomes higher than 85\% if also medium hubs are included. These data are in line with the ones published by
\cite{FAA}. \\
Furthermore, the degrees of the nodes of the
network highlight that each U.S. airport is connected on average to
23 airports, unless large hubs are connected to more than 200
airports.

If we focus only on domestic market, we have roughly 740 millions of
passengers. In this case, the network is characterized by 1149 U.S.
airport and 20445 connections between them. As shown by the strength
distribution (Figure \ref{F:DSdist}, right side), a significant
proportion of traffic (around 85\%) is concentrated around the top
61 airports, considering both large and medium hubs.

For the sake of brevity, we do not report a graphical representation
of the strength distributions for in and out-flows. However, it is
worth mentioning that, for both indicators, in and out results are
strongly correlated with the total strength distribution. In other
words, except for some specific airports, we observe similar
patterns between the number of passengers departing from and
arriving at each airport.

In order to compute the local $l$-adjacency clustering coefficients,
we consider only the U.S. domestic market network, preventing
possible distorted effects due to international flights. Indeed,
data regarding connections between airports located outside of U.S.
territory are not included in the dataset, so that if we include
these airports we are not able to effectively catch the presence of
triangles. Notice that the restriction to the domestic flights does
not lead to a noticeable bias of the analysis of the U.S. airport
network as a whole, since domestic market covers roughly 80\% of
total passengers that arrive or depart from the U.S. airports in
2017.

Figure \ref{F:Clust} displays the distributions of the components of
the local $l$-adjacency clustering coefficients vector
$\mathbf{c}^{all}(l)$ computed at different levels $l$ and
considering as nodes either all the airports (on the left side) or
only large and medium hubs (on the right side). We also report
synthetic measures in $\mathbf{h}(l)$ for alternative choices of
weights $x$'s.

As a premise, the classical global clustering coefficient, obtained
as mean of the coefficients $c^{all}_{i}(0)$  over the nodes
$i=1,...,n$, is equal to $0.56.$ When referring to the local
$l$-adjacency clustering coefficients, we notice that the
distribution shows high volatility, enhancing relevant differences
between airports, and negative skewness, showing a median equal to
0.67, significantly greater than the mean.


\begin{figure}[!h]
    \centering
    \includegraphics[scale=0.3]{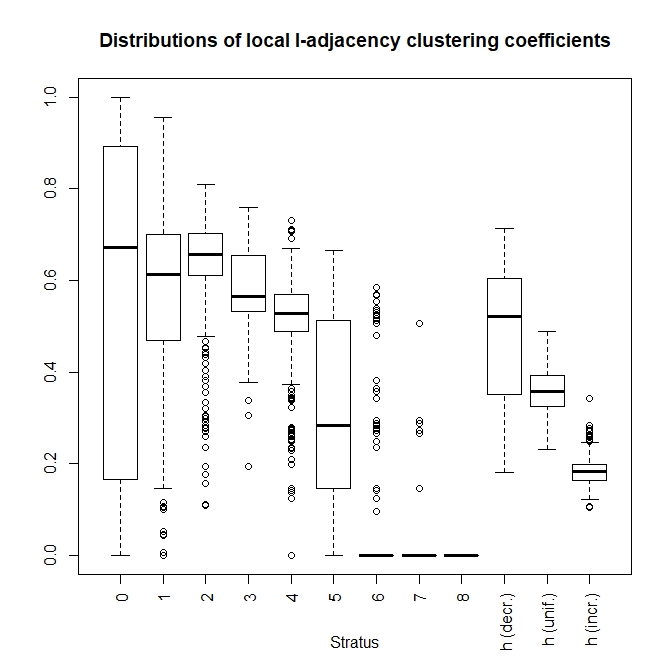}
    \includegraphics[scale=0.3]{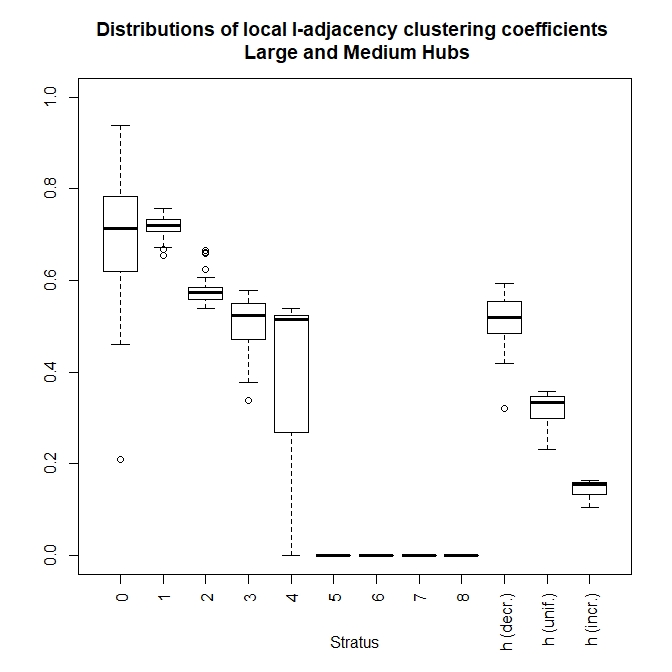}
    \caption{Distributions of the elements of the vector of the $l$-adjacency clustering coefficients
$\mathbf{c}^{all}(l)$ computed at different levels $l$. Furthermore,
the distributions of the elements of $\mathbf{h}^{all}$ are reported
for three different choices of weights (decreasing, uniform and
increasing respectively). On the right side, the same distributions
are computed considering only large and medium hubs.}
    \label{F:Clust}
\end{figure}

Focusing on large and medium airports (Figure \ref{F:Clust}, right
side), the average clustering increases, as the mean is equal to
0.69 and the median is equal to 0.71.
Except for the Ted Stevens Anchorage
International\footnote{$c^{all}_{i}(0)$ is equal to 0.21 for this
airport.}, a medium hub located in Alaska, all relevant airports in
terms of passengers traffic have a clustering coefficient not lower
than 0.5. The different behavior of the Anchorage airport can be
easily justified by the specific characteristics of this hub. The
airport is indeed connected to strategic hubs and to some other
remote airports in U.S. as well. Among larger hubs, the highest
rankings are instead observed for Ontario International (CA) and
Southwest Florida International (FL). Both these airports are
characterized by a low proportion of direct connections, but they
are on average connected to airports that are connected each other.

Different patterns, in terms of classical clustering coefficient,
between larger and smaller hubs can be partially explained also by
the number of geodesic paths moving from the nodes and with a given
length. To this aim we refer to Figure \ref{F:Dist}, which reports
the percentage of the geodesic paths of a fixed length (vertical
axis) versus the total strength (horizontal axis) for all the nodes
of the network.

On upper left-side, the figure depicts the proportion of geodesics
of length 1 for each airport. Large and medium hubs are on average
directly connected to 15\% and 11\% of total airports, while smaller
hubs are directly connected to 1\% of the total nodes.

\begin{figure}[!h]
    \centering
    \includegraphics[scale=0.55]{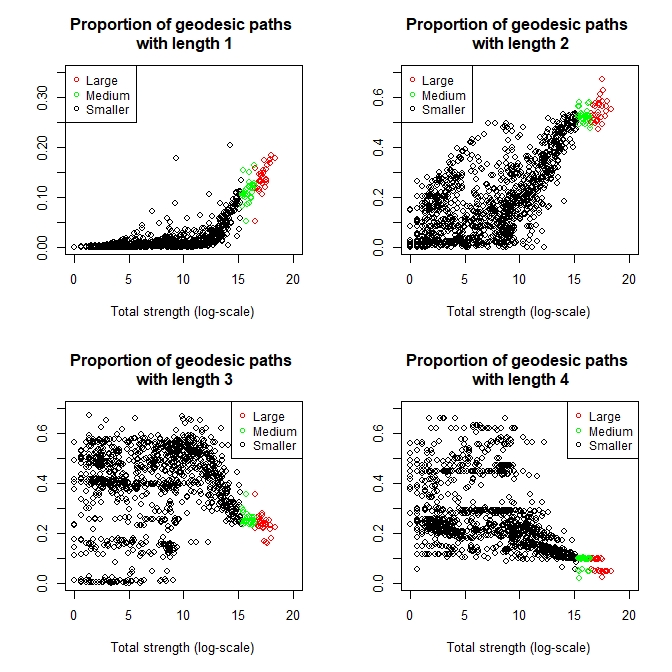}
    \caption{Proportion of geodesic paths with different lengths ($1,2,3,4$, respectively)
vs total strength (reported on log-scale) of each airport. Large and
medium hubs are reported in red and green respectively.}
    \label{F:Dist}
\end{figure}

Moving to the analysis of the local $l$-adjacency clustering
coefficient for node $i$ when $l=1$ of the type $c^{all}_{i}(1)$,
such a coefficient seizes possible connections of the $i-th$ airport
with high clustered areas, reachable from it with one stopover. In
this case, we observe a slight reduction of average clustering and a
general decrease of the variability between different airports.
Large and medium hubs have instead a different behavior, showing an
average increase of the $1$-adjacency clustering coefficients (mean
and median move respectively from 0.69 and 0.71 to 0.72 and 0.73
respectively). It is worth mentioning the considerably low
volatility of the distribution of the components of
$\mathbf{c}^{all}(1)$ for these hubs. The elements of the
$\mathbf{c}^{all}(1)$ range indeed in the interval $(0.65-0.76)$.

Results show that larger hubs are highly clustered but, according to
the local $l$-adjacency clustering coefficient with $l=1$, they are
also directly connected to strong communities, confirming their
strategic role in the airport system.

Focusing on higher levels in terms of distances (see Figure
\ref{F:Dist}, upper right side), we detect a significant proportion
of geodesics of length 2 in the network. On average, large and
medium hubs are respectively connected to 55\% and 52\% of total
airports via geodesic paths of length 2. Through these paths they
reach strong communities as well as non-primary hubs characterized
by a low clustering coefficient. Hence,  $c^{all}_{i}(2)$ is lower
than $c^{all}_{i}(1)$ for all airports of label $i$ belonging to
this group, with reductions that vary between 2\% and 38\%. Smaller
hubs reach instead 20\% of the nodes in two steps. Typically, they
are connected to high clustered areas showing a $c^{all}_{i}(2)$
higher than $c^{all}_{i}(1)$ and $c^{all}_{i}(0)$.

As regard to higher strati, we observe in Figure \ref{F:Dist} a
proportion of geodesics of lengths 3 and 4, equal to 40\% and 27\%
for small hubs, respectively. Large and medium hubs reach instead
25\% and 8\% of airports in 3 and 4 steps, respectively. As a
consequence, the local $l$-adjacency clustering coefficient is
slowly decreasing with respect to $l$ for small hubs, while an
higher reduction is observed for larger hubs. For the latter
category, it is worth noting the high volatility of the components
of $\mathbf{c}^{all}(4)$. In particular, roughly an half of relevant
hubs has a value of such clustering coefficient higher than 0.5.
Seattle-Tacoma International (WA), Ted Stevens Anchorage
International (AK), Daniel K. Inouye International (HI) and Kahului
(HI) airports show instead very low local $4$-adjacency clustering
coefficients ($\leq 0.2$), mainly justified by the fact that
geodesics of length 4 usually connect remote airports with a weak
community structure at stratus 0.

In the line with what evidenced with the case $l=4$, one can notice
that, on average, small airports are connected to the 11\% of total
airports by geodesics of length 5. However, a very high volatility
is observed in this class of airports. Some specific non-primary
airports are able to reach more than an half of the airports through
geodesics of length 5. These patterns justify the significant
volatility and a not negligible average of the elements of
$\mathbf{c}^{all}(5)$. Larger hubs have instead very few connections
at this stratus ($< 1\%$) leading to a clustering coefficient close
to zero. This argument is confirmed and furtherly stressed for
strati greater than 5. Only few nonprimary hubs $i$ are connected to
some other nodes through geodesic paths with length larger than 5,
hence showing values of $c^{all}_{i}(l)$ greater than zero for
$l>5$. Indeed, typically, these connections regard relations between
very remote and without a strong community structure airports. For
instance, Blakely Island (WA) and Tatitlek (AK) airports are
connected by a geodesic path of length 8\footnote{The geodesic path
is given by the following sequence of edges: Blakely Island --
Friday -- Kenmore Air Arbor -- Roche Harbor Country -- Seattle
Tacoma International airport -- Ted Stevens Anchorage International
Airport Country -- Beluga airport -- Merrill Field Anchorage Airport
-- Tatitlek airport}.

The values of the elements of $\mathbf{h}^{all}$ in Figure
\ref{F:Clust} synthesizes the overall community structure of each
node, thus providing a measure of the relevance of the node in the
network. The choice of weights $x_l$ can modulate the intensity of
the elements of $\mathbf{c}(l)$ in contributing to the overall
stratified community structure, giving to this indicator a high
degree of flexibility.

Here, we consider the three possible scenarios already used in
Section \ref{sec:example} (i.e. decreasing, uniform and increasing
weights). For instance, assuming that weights $x$'s are decreasing,
we are reducing the impact of the elements of local $l$-adjacency
clustering coefficients $\mathbf{c}(l)$ with respect to the whole
system when the distance $l$ increases. In particular, by
concentrating the mass of weights over the small values of $l$, we
take into major consideration the community structures close to the
nodes of the graph. In this case, the average of the components of
$\mathbf{h}^{all}$ is equal to 0.48, and it is higher for large and
medium hubs (equal to 0.54). Therefore, large and medium airports
are confirmed to be strategic hubs in the network. Indeed, on one
hand, these airports are involved in strong communities at low
strati; on the other hand, they are directed connected to high
clustered areas.

Differently, the cases of either uniform weights or concentration of
the $x$'s over large values of $l$ emphasize the relevant role of
peripheral communities of airports. Since these scenarios are more
sensitive to communities far from the nodes, we observe a reduction
of the average of the components of $\mathbf{h}^{all}$, equal to
0.36 (uniform weights) and 0.18 (increasing weights), and higher
values for smaller hubs.

We focus now on computing the in and out local $l$-adjacency
clustering coefficients by means of a separate evaluation of in- and
out-paths. As stressed before, the airport network is highly
symmetric so that it is usually analysed as an undirected one (see
\cite{Barrat_2004})). In our case, we observe a strong positive
correlation (close to 1) between in- and out-degree (and between in
and out-strength). To assess the symmetry of the network, Fagiolo
(\cite{Fagiolo_2006}) proposes a specific measure $S$. If the value
of $S$ is close to zero, then an empirically-observed network is
sufficiently symmetric to justify an undirected network analysis. In
our case, we obtain 0.02. This index becomes 0.19 when weights are removed. Hence, network is weakly asymmetric with a
more pronounced behaviour when weights are not considered. However, although
direct connections are highly correlated, some differences could be
observed when we focus on long geodesics.

As a consequence, at stratus $l=0$, in- and out- local $l$-adjacency
clustering coefficients are very similar (see Figure
\ref{F:ClustInOut}) and lower than $c^{all}_{0}$. With the directed
(in and out) local $l$-adjacency clustering coefficients, we are
focusing on specific patterns. Indeed we neglect some types of
triangles (like cycles and middleman triangles, according to the
classification provided in \cite{Fagiolo_2007}) and consider only
directed paths.

On average we observe slightly higher out-clustering coefficients
for larger hubs and lower ones for smaller airports. However, there
is not an univocal pattern among the airports, although in many
cases differences are negligible. In the class of medium hubs, an
interesting case is the node $i$ associated to the Luiz Munoz Marin
International Airport in San Juan (Puerto Rico), characterized by a
$c^{out}_{i}(0)$ equal to 0.81 against a $c^{in}_{i}(0)$ equal to
0.77. In other words, this airport is more involved in weighted
triangles of out-type than in-triangles. This evidence is partially
justified by a number of passengers departure higher than arrivals,
probably motivated by higher movements towards U.S. than vice versa.

\begin{figure}[!h]
    \centering
    \includegraphics[scale=0.3]{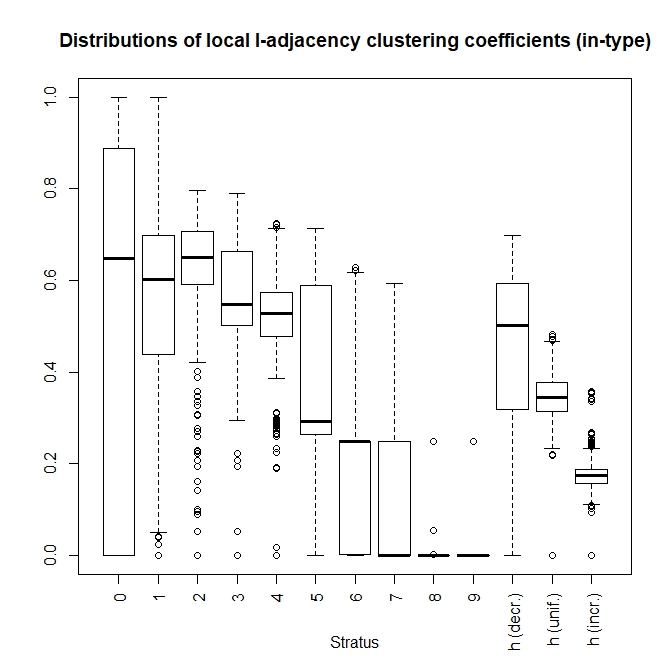}
    \includegraphics[scale=0.3]{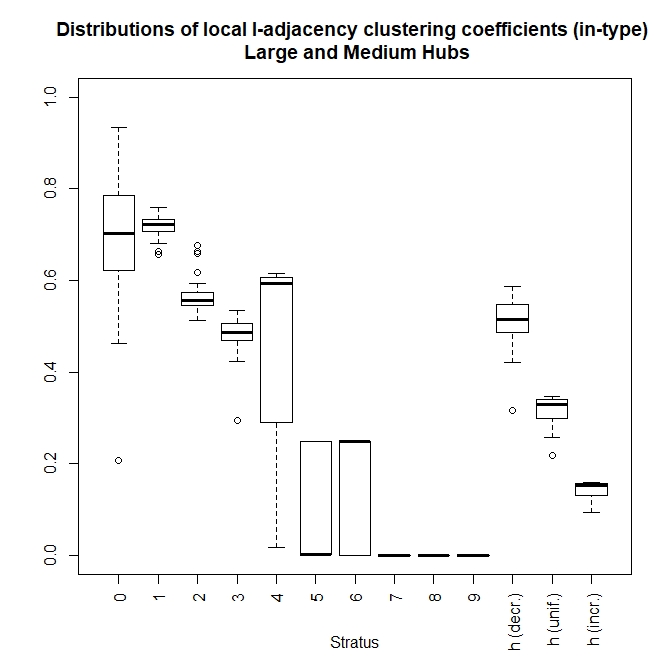}
    \includegraphics[scale=0.3]{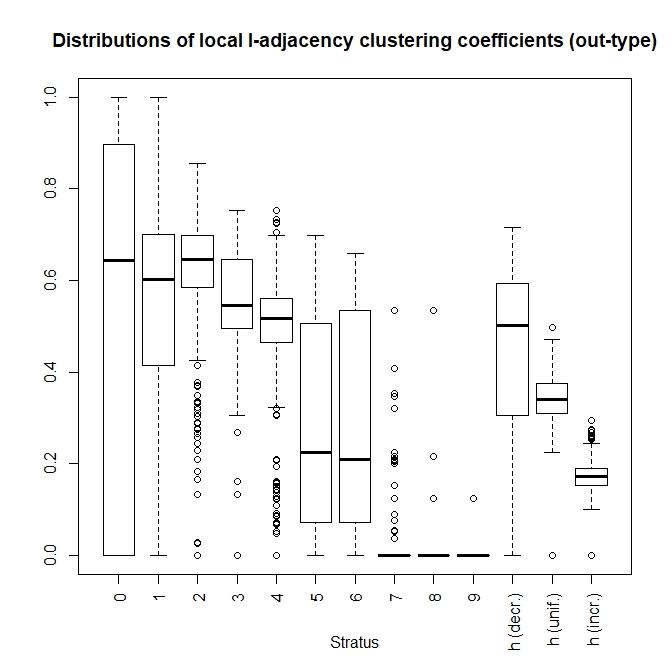}
    \includegraphics[scale=0.3]{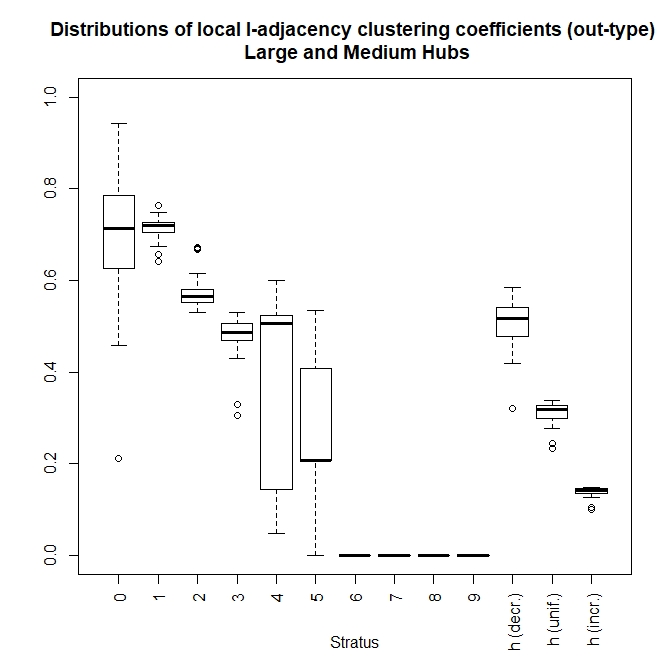}
\caption{Upper left figure displays the distributions of local
$l$-adjacency clustering coefficients of in-type in
$\mathbf{c}^{in}(l)$ and distributions of the elements of
$\mathbf{h}^{in}$ for the considered different choices of weights
(decreasing, uniform and increasing respectively). On the right
side, the same distributions are computed by considering only large
and medium hubs. In the bottom part, we report the distributions of
the components of $\mathbf{c}^{out}(l)$ and $\mathbf{h}^{out}$ for
all the airports (left side) and only for large and medium hubs
(right side).}
    \label{F:ClustInOut}
\end{figure}

The directed $l$-adjacency clustering coefficients display similar
distributions and, on average, lower values than the elements of
$\mathbf{c}^{all}(l)$ until $l$ assumes values equal to 4.
Remarkable differences are instead observed for higher strati. In
particular, stronger communities of in-type are observed in
peripheral nodes. Since stratus 5, we observe indeed higher values
of the components of $\mathbf{c}^{in}(l)$ than
$\mathbf{c}^{out}(l)$.

The analysis of the synthetic indicators given by the components of
$\mathbf{h}^{in}$ and $\mathbf{h}^{out}$, when decreasing weights
are considered, confirms the relevance of large and medium hubs in
terms of community structures (both in and out)\footnote{This
evidence is a consequence of the high correlation between in and out
$l$-adjacency correlation coefficients computed for low values of
$l$ (at this regard, see Figure \ref{F:Corr})}. Instead, when we
base our analysis on increasing weights, as expected the role of
peripheral nodes is emphasized. In this case a lower correlation
(see Figure \ref{F:Corr}) is observed depending on the specific
behavior of each airport. Furthermore, although the network is only
weakly asymmetric in terms of adjacency matrix, different patterns of long directed
paths are also caught by the synthetic measure. We have indeed a
slight prevalence of community structures of in-type.

\begin{figure}[!h]
    \centering
    \includegraphics[scale=0.4]{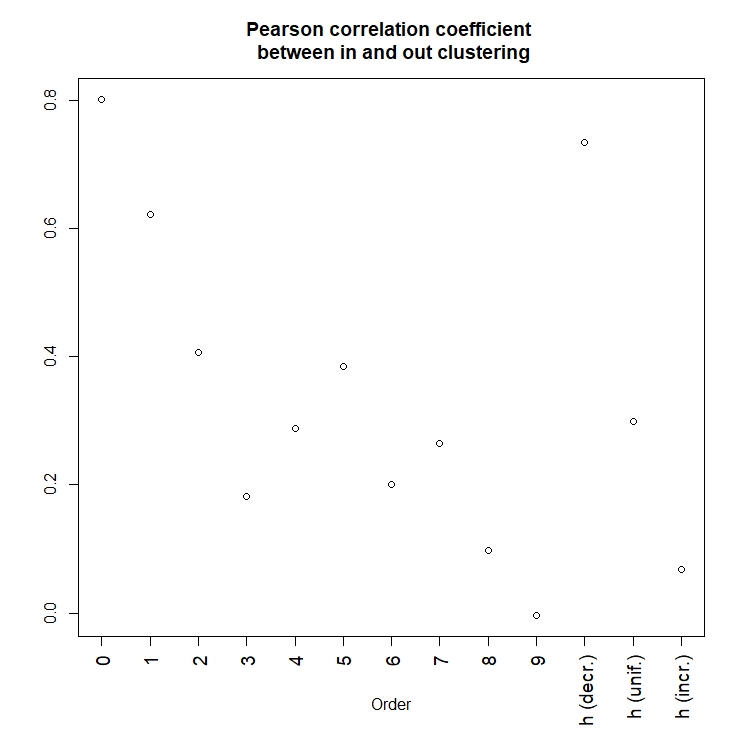}
\caption{Pearson correlation coefficient between in and out
$l$-adjacency correlation coefficients at different strati $l$.}
    \label{F:Corr}
\end{figure}

\newpage

\bigskip
\section{Conclusions}\label{Concl}
Interconnection plays a fundamental role in the business research
context. As well-known, in network theory, the level of
interconnectivity in the neighbourhood of a node is typically
assessed by means of the clustering coefficient and captures the
community structure associated to the considered node. Moving from
this fact, we exploit the concept of community structure to
understand the contextualization of the nodes within the overall
system. In particular, we provide a generalization of the concept of
clustering coefficient in order to catch both the presence of
clustered areas around a node and/or high levels of mutual
interconnections at different distances from the node itself. With
respect to classical clustering coefficient, we are able to capture
in a better way the topological structure of the whole network and
to map the presence of stratified community structures in the
network at different levels. Furthermore, we also define a synthetic
indicator for each node in order to simultaneously consider all the
coefficients. Being this indicator dependent on a set of weights of
the strati, we allow for a degree of flexibility in order to
modulate the effects of both adjacent nodes and peripheral nodes.

An empirical application to U.S. domestic air traffic network is
developed. Results show the effectiveness of these measures in
catching the peculiar characteristics of different nodes in the
airport network. In particular, focusing on large and medium hubs,
we are able to emphasize their strategic role in the airport system.
We observe, indeed, that larger hubs are not only highly clustered
but also at the center of strong communities. When different
communities strati are analysed, the effects of indirect connections
with remote airports are emphasized. Finally, although the network
is only weakly asymmetric and it is typically analysed as an
undirected network in the literature, we show that a separate
evaluation of directed paths at high levels can be useful to
identify specific patterns in terms of in and out-communities.

\bibliography{Myref}
\newpage

\appendix
\begin{table}[!h]
    \footnotesize
    \begin{tabular}{ c  c  c}
        \hline
        \hline
        Rank & City & Airport Name \\
        \hline
        \hline
        1 & Atlanta & Hartsfield - Jackson Atlanta International \\
        2 & Los Angeles & Los Angeles International \\
        3 & Chicago & Chicago O'Hare International \\
        4 & Fort Worth & Dallas-Fort Worth International\\
        5 & Denver & Denver International \\
        6 & New York & John F Kennedy International\\
        7 & San Francisco & San Francisco International\\
        8 & Las Vegas & McCarran International \\
        9 & Seattle & Seattle-Tacoma International\\
        10 & Charlotte  & Charlotte/Douglas International \\
        11 & Newark & Newark Liberty International \\
        12 & Orlando & Orlando International \\
        13 & Phoenix & Phoenix Sky Harbor International \\
        14 & Miami & Miami International \\
        15 & Houston & George Bush Intercontinental/Houston \\
        16 & Boston & General Edward Lawrence Logan International \\
        17 & Minneapolis & Minneapolis-St. Paul International/Wold-Chamberlain \\
        18 & Detroit & Detroit Metropolitan Wayne County \\
        19 & Fort Lauderdale & Fort Lauderdale/Hollywood International \\
        20 & New York & Laguardia \\
        21 & Philadelphia & Philadelphia International \\
        22 & Glen Burnie & Baltimore/Washington International Thurgood Marshall \\
        23 & Salt Lake City & Salt Lake City International \\
        24 & Arlington & Ronald Reagan Washington National \\
        25 & San Diego & San Diego International \\
        26 & Dulles & Washington Dulles International \\
        27 & Chicago & Chicago Midway International \\
        28 & Honolulu & Daniel K. Inouye International \\
        29 & Tampa & Tampa International \\
        30 & Portland & Portland International \\
        \hline
        \hline
    \end{tabular}
    \caption{Large Primary Hubs according to FAA classification (based on enplanements in 2017)}
    \label{tab:CC3}
\end{table}

\begin{table}[!h]
    \footnotesize
    \begin{tabular}{ c  c  c}
        \hline
        \hline
        Rank & City & Airport Name \\
        \hline
        \hline

        31 & Dallas & Dallas Love Field \\
        32 & St. Louis & St Louis Lambert International \\
        33 & Nashville & Nashville International \\
        34 & Austin & Austin-Bergstrom International\\
        35 & Houston & William P. Hobby \\
        36 & Oakland & Metropolitan Oakland International \\
        37 & San Jose & Norman Y. Mineta San Jose International \\
        38 & Metairie & Louis Armstrong New Orleans International \\
        39 & Raleigh & Raleigh-Durham International \\
        40 & Kansas City & Kansas City International \\
        41 & Sacramento & Sacramento International \\
        42 & Santa Ana & John Wayne Airport-Orange County \\
        43 & Cleveland & Cleveland-Hopkins International \\
        44 & San Antonio & San Antonio International \\
        45 & Fort Myers & Southwest Florida International\\
        46 & Indianapolis & Indianapolis International \\
        47 & Pittsburgh & Pittsburgh International \\
        48 & San Juan & Luis Munoz Marin International \\
        49 & Greater Cincinnati & Cincinnati/Northern Kentucky International \\
        50 & Columbus & John Glenn Columbus International \\
        51 & Kahului & Kahului \\
        52 & Milwaukee & General Mitchell International \\
        53 & Windsor Locks & Bradley International \\
        54 & West Palm Beach & Palm Beach International\\
        55 & Jacksonville & Jacksonville International \\
        56 & Anchorage & Ted Stevens Anchorage International \\
        57 & Albuquerque & Albuquerque International Sunport \\
        58 & Burbank & Bob Hope \\
        59 & Buffalo & Buffalo Niagara International \\
        60 & Ontario & Ontario International \\
        61 & Omaha & Eppley Airfield \\
        \hline
        \hline
    \end{tabular}
    \caption{Medium Primary Hubs according to FAA classification (based on enplanements in 2017)}
    \label{tab:CC3b}
\end{table}

\end{document}